\documentclass{amsart}
\usepackage{amsmath,latexsym, graphicx}
\usepackage[psamsfonts]{amssymb}
\usepackage{times}
\usepackage[mathcal]{euscript}
\usepackage{color}
\usepackage{amsfonts}
\usepackage{amssymb}
\usepackage{amsmath}
\usepackage{amsthm}
\usepackage{bm}
\usepackage[english]{babel}
\usepackage[bookmarks]{hyperref}
\numberwithin{equation}{section} \textwidth=140mm \textheight=200mm
\renewcommand{\epsilon}{\varepsilon}
\renewcommand{\epsilon}{\varepsilon}

\renewcommand{\d}{\mathrm{d}}
\renewcommand{\hat}{\widehat }

\newcommand{\ess}{\mathrm{ess}}

\newcommand{\black}{\color{black}}

\newcommand{\be}{\begin{equation}}
\newcommand{\ee}{\end{equation}}

\newcommand{\R}{\mathbb{R}}

\newcommand{\T}{\mathbb{T}}

\renewcommand{\H}{\mathbb{H}}
\newcommand{\V}{\mathbb{V}}

\newcommand{\Z}{\mathbb{Z}}
\newcommand{\cC}{{\mathcal C}}

\newcommand{\cE}{{\mathcal E}}
\newcommand{\cF}{{\mathcal F}}

{\bf}{\it}
\newtheorem{theorem}{Theorem}[section]
\newtheorem{lemma}[theorem]{Lemma}
\newtheorem{corollary}[theorem]{Corollary}

\newtheorem{proposition}[theorem]{Proposition}

\date{\today}
\begin{document}

\title[The number and location of eigenvalues]{The number and location of eigenvalues for the two-particle Schr\"odinger
operators on lattices}

\author{Saidakhmat N. Lakaev,  Mukhayyo O. Akhmadova}

\address[Saidakhmat Lakaev]{Samarkand State University, 140104, Samarkand, Uzbekistan}
\email{slakaev@mail.ru}

\address[Muhayyo Akhmadova]{Samarkand State University, 140104, Samarkand, Uzbekistan}
\email{mukhayyo@mail.ru}

\begin{abstract}
We study the Schr\"odinger operators $H_{\gamma \lambda  \mu}(K)$, $K\in\T$ being a fixed (quasi)momentum of the particles pair, associated with a system of two identical bosons on the one-dimensional lattice $\mathbb{Z}$, where the real quantities $\gamma$, $\lambda$ and $\mu$ describe the interactions between pairs of particles on one site, two nearest neighboring sites and next two  neighboring sites, respectively. We found a partition of the three-dimensional space $(\gamma, \lambda,\mu)$ of interaction parameters into connected components and the exact number of eigenvalues of this operator that lie below and above the essential spectrum, in each component. Moreover, we show that for any $K\in\T^d$  the number of eigenvalues of $H_{\gamma\lambda\mu}(K)$ is not less than the corresponding number of eigenvalues of $H_{\gamma\lambda\mu}(0)$. 
\end{abstract}

\maketitle
\section{\textbf{Introduction}}\label{sec:intro}
Lattice models of physical systems are one of the widely used mathematical models in mathematical physics. Few-body Hamiltonians
\cite{Mattis:1986}, among such models may be viewed as the simplest version of
the corresponding Bose-Hubbard model involving a finite number of particles of a certain type. The few-body lattice Hamiltonians have been intensively studied over the past several decades \cite{ALzM:2004,ALMM:2006, ALKh:2012,BdsPL:2017,FICarroll:2002, HMumK:2020, KhLA:2022, Lakaev:1993, LAbdukhakimov:2020, LBA:2022, LBozorov:2009, LKhKh:2021,LKhL:2012, LO'zdemir:2016, Motovilov:2001}. 

Another reason for studying the discrete Hamiltonians is that they can provide a natural approximation for their continuous counterparts \cite{Faddeev:1986}, which allows studying few-body
systems. It is well known that the Efimov effect \cite{Efimov:1970} was originally attributed to the
three-body systems moving in the three-dimensional continuous space $\mathbb{R}^3$. A rigorous
mathematical proof of the Efimov effect has been given in
\cite{Ovchinnikov:1979,Sobolev:1993,Tamura:1991,Yafaev:1974}. The celebrated Efimov phenomenon to take place also in the lattice three-particle  systems \cite{ALzM:2004, ALKh:2012,DzMSh:2011,Lakaev:1993}. 
Discrete
Schr\"odinger operators also represent a simple model
for description of few-body systems formed by particles traveling
through periodic structures, such as ultracold atoms injected into
optical crystals created by the interference of counter-propagating
laser beams \cite{Bloch:2005, Winkler:2006}. The study of
ultracold few-atom systems in optical lattices have became popular
in the last years due to availability of controllable
parameters, such as temperature, particle masses, interaction potentials etc. (see e.g.,
\cite{Bloch:2005,JBC:1998, JZ:2005, Lewenstein:2012}
and references therein). These possibilities give an opportunity to experimentally observe stable repulsively bound pairs of ultracold atoms (\cite{Ospelkaus:2006,Winkler:2006}, which is not the case for the traditional condensed matter
systems, where stable composite objects are formed by means attractive forces. Lattice Hamiltonians are of particular interest in fusion physics too. For example, in \cite{Motovilov:2001}, a
one-dimensional one-particle lattice Hamiltonian has been successfully employed
to show that certain class of molecules in lattice structures may enhance nuclear fusion probability.

In the continuous case, the center-of-mass motion can be separated, which is not the case for lattice few-body problems. However, for the lattice Hamiltonian $\mathrm{H}$ acting in the functional Hilbert space $\T^{n \cdot d} 
$ we have a von Neumann direct integral decomposition
\begin{equation}\label{entries}
\mathrm{H}\simeq\int\limits_{K\in \T^d}
^\oplus  H(K)\,d K,
\end{equation}
where $\T^d$ is the $d$-dimensional torus.
The so called fiber Hamiltonians $H(K)$ acting on the Hilbert space $\T^{(n-1)d}$ nontrivially depends on the quasimomentum $K\in\T^d$ (see e.g., \cite{ ALzM:2004, ALMM:2006, RSimonIII:1982}). This decomposition allows us to reduce the problem to studying the operators $H(K)$. 

In this work, we study the spectral properties of the fiber Hamiltonians $H(K),\,K\in\T$ acting in the Hilbert space $L^{2,e}(\T)$ as
\begin{equation}\label{fiber}
H_{\gamma\lambda\mu}(K):=H_0(K) + V_{\gamma\lambda\mu},
\end{equation}
where $H_0(K)$ is the fiber kinetic-energy operator, 
$$
\bigl(H_0(K)f\bigr)(p)=\cE_K(p)f(p), 
$$
with
\begin{equation}\label{def:dispersion}
\cE_K(p):= 2(1-\cos\tfrac{K}2\,\cos p)
\end{equation}
and 
$V_{\gamma\lambda\mu}$ is the combined interaction potential.\black 
The parameters  $\gamma$, $\lambda$ and $\mu$, called coupling constants, describe interactions between the particles which are located on one site,  on the nearest neighboring sites of the lattice and in the next nearest neighboring sites, respectively.

Within this new model, we find both the exact number of eigenvalues and their locations of the operator $H_{\gamma\lambda\mu}(0)$. We describe the mechanisms of emission and absorption of the eigenvalues of $H_{\gamma\lambda\mu}(0)$ at the thresholds of its essential (continuous) spectrum depending on the interaction parameters. Furthermore,  we establish sharp lower bounds for the number of isolated eigenvalues $H_{\gamma\lambda\mu}(K)$ depending on the quasi-momentum $K\in\T$, which lie both below the essential spectrum and above that. 

For this, we apply the results obtained for the operator $H_{\gamma\lambda\mu}(0)$ and the nontrivial dependence of the dispersion relation $\cE_K$ on the (quasi)momentum $K\in\T^q$.

We recall that the two-particle Schr\"{o}dinger operator ${H}_{\mu}(K)=-\Delta+\mu {V},\,\mu>0$ on the lattice $\Z^d$ associated to a system of two bosons with zero-range repulsive interactions $\mu>0$ has been considered as a theoretical
basis for explanation of the experimental results obtained in \cite{Ospelkaus:2006, Winkler:2006}.

Note that the continuous counterpart of the two-particle Schr\"{o}dinger operators on lattices, which associated with a system of two quantum-mechanical particles on $\R^d$ interacting via short-range potential $v(x)$ has isolated eigenvalues lying only  below the essential spectrum fulfilling the semi-axis $[0, + \infty)$ and hence this model is well adapted to describe systems of two-particles with attractive interactions.

To study the eigenvalues of the discrete Schr\"odinger operator $H_{\gamma\lambda\mu}(K)$, we apply analytic function theory, namely, we investigate the corresponding Fredholm determinant $\Delta_{\gamma\lambda\mu}(K,z)$, as there is a one-to-one mapping between the sets of eigenvalues of the operator $H_{\gamma\lambda\mu}(K)$ and the zeros of  $\Delta_{\gamma\lambda\mu}(K,z)$ (see \cite{LKhKh:2021}). Correspondingly, the change in the number of zeros of Fredholm determinant $\Delta_{\gamma\lambda\mu}(0,z)$ results in the change of the number of isolated eigenvalues of Schr\"odinger operator $H_{\gamma\lambda\mu}(0)$.

Our main finding is that the number of zeros of the determinant $\Delta_{\gamma\lambda\mu}(0,z)$ located below (resp.  above) the essential spectrum changes if and only if the principal term $C^{-}(\gamma,\lambda, \mu)$ (resp. $C^{+} (\gamma, \lambda, \mu)$) of the asymptotics of the Fredholm determinant $\Delta_{\gamma\lambda\mu} (0,z)$ vanishes as $z$ approaches the lower (resp. upper) threshold of the essential spectrum (see Lemma \ref{lemm:asimpdeter}).
 
Using this property, we establish a partition of the three-dimensional $(\gamma, \lambda,\mu)$-space into four disjoint connected components by means of curves $C^{-}(\gamma,\lambda, \mu)=0$ or
$C^{+}(\gamma,\lambda,\mu)=0$. This allows us to prove that the number of zeros of the Fredholm determinant is constant in each connected component.  

In \cite{LKhKh:2021,LO'zdemir:2016}, it was  studied the Schrödinger operators on the lattice $\Z^d,\,d=1,2$, associated to a system of two bosons with the zero-range on one site interaction ($\lambda\in \R$)  and  interactions on the nearest neighboring sites ($\mu\in \R$) of the $d$- dimensional lattice $\Z^d$.
We emphasize that, our results is an extensions of \cite{LKhKh:2021, LO'zdemir:2016}. The authors of \cite{LKhKh:2021} consider the Schr\"{o}dinger operators  $H_{\lambda\mu}(K)$ on two-dimensional lattice $\Z^2$. The operator $H_{\lambda\mu}(K)$ can have one or two  eigenvalues, lying as below the essential spectrum, as well as above it. The connected components,  which split the two-dimensional  plane $\R^2$ of interaction parameters are described by means of a second order elementary curves (hyperbolas). Similar results for the number of eigenvalues of one-particle Schr\"odinger operators in $\Z^d,\,d\ge1$ have been obtained, for instance in \cite{LBozorov:2009} with attractive interactions and  $d=3$ and also with attractive and repulsive interaction cases in \cite{HMumK:2020} for all $d\ge1$ considering only negative eigenvalues.
 
In the present work, the connected components are described by implicit functions (third-order polynomials) of three variables defined on $\R^3$. This is a more challenging case as the connected components are described by the third-order surfaces.

The discrete two-particle Schr\"{o}dinger operator
${H}_{\mu}(k)$ associated to a system of one and two
quantum-mechanical particles on $\Z^d$ interacting via short-range
potentials have been studied in recent years \cite{ALMM:2006, BdsPL:2017,KhLA:2021, LAbdukhakimov:2020, Motovilov:2001}.

Note that some results such as the existence of eigenvalues and their finiteness can be received for a large class of Schr\"odinger type operators (see e.g., \cite{ Klaus:1977, KSimon:1980, KhLA:2021,LKAlladustov:2021}).
However, our results show that the study of a qualitative change  in the number of eigenvalues of $ H_{\gamma\mu\lambda}(K),$ even
for $K=0,$ is very delicate: There is ball, with arbitrarily small radius, in the three-dimensional $(\gamma,\lambda, \mu)$-space in which the number of eigenvalues has a jump (see Theorem \ref{teo:number_K=0}).

The paper is organized as follows: The section \ref{sec:intro} is introduction. In Section \ref{sec:hamiltonian},
we introduce the two-particle Hamiltonian $\H_{\gamma\lambda\mu}$ of a system of
two bosons in the position and momentum space representations and also the Schr\"odinger operator $H_{\gamma\lambda\mu}(K)$ associated to the
Hamiltonian $\H_{\gamma\lambda\mu}$. The main results of the paper are stated in
Section \ref{sec:main_results} and their proofs are contained in Section \ref{sec:proof}. 

\black
\section{The two-particle Hamiltonian} \label{sec:hamiltonian}

\subsection{The position-space representation}\label{subsec:position}

Let  $\Z$ be the one-dimensional lattice and $\Z \times \Z$ be cartesian square of $\Z$. 
Let  $\ell^{2,s}(\Z \times \Z)$ be the Hilbert space of square-summable
symmetric functions on $\Z \times \Z$. 

The free Hamiltonian $\hat \H_0$ of a system of two identical particles (bosons), in the position space representation, is usually associated with the following self-adjoint (bounded) Toeplitz-type operator on the Hilbert space $\ell^{2,s}(\Z \times \Z)$:

\begin{equation*}
(\hat \H_0 \hat f)(x_1,x_2)= \sum_{s_1\in\Z}  \hat \epsilon(x_1-s_1)\hat f(s_1,x_2) +\sum_{s_2\in\Z}\hat \epsilon(x_2-s_2) \hat f(x_1,s_2), \,\, \hat f \in \ell^{2,s}(\Z \times \Z)
\end{equation*}

where
\begin{equation}\label{def:epsilon}
\hat \epsilon(s) =
\begin{cases}
2 & \text{if $|s|=0,$}\\
-\frac{1}2 & \text{if $|s|=1,$}\\
0 & \text{if $|s|>1.$}
\end{cases}
\end{equation}

The interaction operator $\hat \V_{\gamma\lambda\mu}$, in the position space representation, is the multiplication
operator by the function $\hat v \in \ell^{1}(\Z),   i.e.,$

\begin{equation*}\label{interaction}
\hat \V_{\gamma\lambda\mu} \hat f (x,y) = \hat v_{\gamma\lambda\mu}(x-y) \hat
f(x,y),\,\,\hat f \in \ell^{2,s}(\Z \times \Z),
\end{equation*}

where
\begin{equation}\label{def:potentials}
\hat v_{\gamma\lambda\mu}(s)=
\begin{cases}
2\gamma & \text{if $|s|=0,$}\\
\lambda & \text{if $|s|=1,$}\\
\mu & \text{if $|s|=2,$}\\
0 & \text{if $|s|>2.$}
\end{cases}
\end{equation}

The total Hamiltonian
$\hat \H_{\gamma\lambda\mu}$ of a system of two identical particles is described as a bounded self-adjoint operator on $\ell^{2,s}(\Z \times \Z)$:
\begin{equation*}\label{two_total}
\hat \H_{\gamma\lambda\mu}=\hat \H_0 + \hat \V_{\gamma\lambda\mu}, \,\, 
\gamma,\lambda,\mu\in\R.
\end{equation*}

\subsection{The two-particle Hamiltonian: the (quasi)momentum-space representation}

Let $\T= \allowbreak  \R /2\pi \Z  \equiv (-\pi,\pi]$ be the
one-dimensional torus, the Pontryagin dual group of $\Z$. Let $L^{2}(\T)$ be
the Hilbert space of square-integrable functions on
$\T$ and let $$ \cF:\ell^2(\Z)\rightarrow L^2(\T), \,\,
\cF \hat f(p)=\frac{1}{2\pi} \sum_{x\in\Z} \hat f(x) e^{ip
x}.
$$
is the
standard Fourier transform with the inverse
$$ \cF^*: L^2(\T) \rightarrow \ell^2(\Z) , \,\,
\cF^*  f(p)=\frac{1}{2\pi} \int_{\T}  f(x) e^{ip
x}dx.
$$

The free Hamiltonian $\H_0=(\cF \otimes \cF)\hat \H_0 (\cF^{*} \otimes \cF^{*})$ of a system of two identical particles,in the momentum space representation, is the multiplication operator by the function $\epsilon(p) + \epsilon(q)$ in the Hilbert space $L^{2,s}(\T \times \T)$ of symmetric functions on the cartesian square $\T \times \T$ of the torus $\T$: 
$$
\H_0 f(p,q) = [\epsilon(p) + \epsilon(q)]f(p,q),
$$
where the continuous function (dispersion relation) $\epsilon$ is given by
$$
\epsilon(p) =[\cF \hat \epsilon](p)= 1-\cos p,\quad
p\in \T.
$$

The interaction operator $\V_{\gamma\lambda\mu}=(\cF \otimes \cF)\hat \V_{\gamma\lambda\mu} (\cF^{*} \otimes \cF^{*})$ is the integral operator of convolution type acting in $L^{2,s}(\T \times \T)$ as
$$
(\V_{\gamma\lambda\mu} f)(p,q) = \frac{1}{2\pi}\int_{\T}
v_{\gamma\lambda\mu}(p-u) f(u,p+q-u)\d u,
$$
where the kernel function $v_{\gamma\lambda\mu}(\cdot)$ is given by
$$
v_{\gamma\lambda\mu}(p)=[\cF \hat v_{\gamma\lambda\mu}] (p)=\frac{1}{2\pi}\sum_{x\in\Z} \hat v_{\gamma\lambda\mu}(x) e^{ip
x} = \frac{1}{\pi} (\gamma+\lambda\cos p+\mu\cos2p),\quad
p\in \T.
$$

The total two-particle Hamiltonian $\H_{\gamma\lambda\mu}$ of a system of two identical quantum-mechanical patricles interacting via a finite range attractive potentials $\hat v_{\gamma\lambda\mu}$, in the momentum space representation, is the bounded self-adjoint operator acting in $L^{2,s}(\T \times \T)$ as 
$$
\H_{\gamma\lambda\mu}:=(\cF\otimes\cF) \hat
\H_{\gamma\lambda\mu}(\cF\otimes\cF)^*=\H_0 + \V_{\gamma\lambda\mu}.
$$

\subsection{The Floquet-Bloch decomposition of $\H_{\gamma\lambda\mu}$ and discrete
Schr\"odinger operators}\label{subsec:von_neuman}

Since the operator $\widehat H_{\gamma\lambda\mu}$ commutes with the shift operators on the lattice $\Z \times \Z$, we can decompose the space $L^{2,s}(\T \times \T)$ and Hamiltonian $\H_{\gamma\lambda\mu}$
into the von Neumann direct integrals as
\begin{equation}\label{hilbertfiber}
L^{2,s}(\T \times \T)\simeq \int\limits_{K\in \T} ^\oplus
L^{2,e}(\T)\,\d K \quad \text{and} \quad \H_{\gamma\lambda\mu} \simeq \int\limits_{K\in \T} ^\oplus
H_{\gamma\lambda\mu}(K)\,\d K,
\end{equation}
respectively, where $L^{2,e}(\T)$ is the Hilbert space of square-integrable even
functions on $\T$ (see, e.g., \cite{ALMM:2006}).

The fiber operator $H_{\gamma\lambda\mu}(K),$ $K\in\T$ is a
self-adjoint operator defined in $L^{2,e}(\T)$ as
\begin{equation*}
H_{\gamma\lambda\mu}(K) := H_0(K) + V_{\gamma\lambda\mu},
\end{equation*}
where the unperturbed operator $H_0(K)$ is the multiplication
operator by the function
$$
\cE_K(p):= 2 (1-\cos\tfrac{K}2\,\cos p)
$$
and the perturbation $V_{\gamma\lambda\mu}$ is defined as
$$
V_{\gamma\lambda\mu} f(p)= \frac{1}{\pi}\int_{\T}  (\gamma+\lambda
\cos p\cos q+\mu\cos 2p\cos 2q) f(q)\d q.
$$

 In some literature, the parameter $K$ is called \emph{quasi-momentum} and the fiber $H_{\gamma\lambda\mu}(K)$
is called \emph{discrete Schr\"odinger operator} associated to
the two-particle Hamiltonian $ \H_{\gamma\lambda\mu}$.

Using the Fourier transform, $\ell^{2,s}(\Z \times \Z)$  and $\H_{\gamma\lambda\mu}$
can also be rewritten as decompositions
$$
\ell^{2,s}(\Z \times \Z)\simeq \int\limits_{K\in \T} ^\oplus
\ell^{2,e}(\Z)\,\d K \quad \text{and} \quad \H_{\gamma\lambda\mu} \simeq \int\limits_{K\in \T} ^\oplus 
H_{\gamma\lambda\mu}(K)\,\d K,
$$
respectively, where $\ell^{2,e}(\Z)$ is the Hilbert space of square-summable
even functions on $\Z$ and
\begin{equation}\label{def}
\hat H_{\gamma\lambda\mu}(K):=\cF^* H_{\gamma\lambda\mu}(K) \cF=\hat H_0(K) + \hat V_{\gamma\lambda\mu},
\end{equation}
and
\begin{equation*}
\hat H_0(K)\hat f(x) \ = \ \sum_{s\in\Z} \hat \cE_K(x-s) 
f(s),\quad \hat f\in\ell^{2,e}(\Z),
\end{equation*}
with
\begin{equation*}
\hat \cE_K(x)=2 \hat\epsilon(x)\cos\frac{K x}{2}
\end{equation*}
and the operator $\hat V_{\gamma\lambda\mu}$ acts in $\ell^{2,e}(\Z)$ as
\begin{equation*}
\hat  V_{\gamma\lambda\mu}\hat f(x)= \hat v_{\gamma\lambda\mu}(x)\hat f(x).
\end{equation*}%

\subsection{The essential spectrum of discrete Schr\"odinger operators} \label{subsec:ess_spec}

Since $V_{\gamma\lambda\mu}$ is a finite rank operator, by well known Weyl's Theorem (see \cite[Theorem XIII.14]{RSimon:IV}) for
any $K\in\T$ the essential spectrum
$\sigma_{\ess}(H_{\gamma\lambda\mu}(K))$ of $H_{\gamma\lambda\mu}$ coincides with the spectrum of $H_0(K),$ i.e.,
\begin{equation}\label{eq:essent_spec}
\sigma_{\ess}(H_{\gamma\lambda\mu}(K))=\sigma(H_0(K)) =
[\cE_{\min}(K),\cE_{\max}(K)],
\end{equation}
where
\begin{align*}
\cE_{\min}(K):= & \min_{p\in  \T}\,\cE_K(p) =
2(1-\cos \tfrac{K}2)\geq 0=\cE_{\min}(0),\\
\cE_{\max}(K):= & \max_{p\in  \T }\,\cE_K(p) =
2(1+\cos \tfrac{K}2)\leq
4=\cE_{\max}(0).
\end{align*}

\black

\section{Main results}\label{sec:main_results}

Let $K\in\T$ and  $n_+({H}_{\gamma\lambda\mu}(K))$ resp. $n_-({H}_{\gamma\lambda\mu}(K))$   be the number of eigenvalues of the operator ${H}_{\gamma\lambda\mu}(K)$ above resp. below its essential spectrum.

Our first main result is a generalization of Theorem 2 in \cite{ALMM:2006}.

\begin{theorem}\label{teo:disc_Kvs0}
Suppose that ${H}_{\gamma\lambda\mu}(0)$ has $n$ eigenvalues below resp.
above the essential spectrum for some $\gamma,\lambda,\mu\in\R.$ Then for
each $K\in\T$ the operator ${H}_{\gamma\lambda\mu}(K)$ has at least $n$
eigenvalues below resp. above its essential  spectrum. 
In other words,
$$
n_-({H}_{\gamma\lambda\mu}(K))\geq n_-({H}_{\gamma\lambda\mu}(0)) 
$$                                            
and  
$$
n_+({H}_{\gamma\lambda\mu}(K))\geq n_+({H}_{\gamma\lambda\mu}(0)).
$$

\end{theorem}

Let us consider the cubic polynomial $C^{\pm}$ of three-variable defined by
\begin{equation} \label{polynomial:c}
C^{\pm}(\gamma,\lambda,\mu):=\mp(\gamma+\lambda+\mu+\gamma\lambda\mu)+\gamma\lambda+2\gamma\mu+\lambda\mu .
\end{equation}
\begin{lemma}\label{lem:function_mu}
The set of points $\mathbb{R}^3$ satisfying the equation $C^{\pm}(\gamma,\lambda,\mu) = 0$ coincides with the graph of function 
\begin{align*}
&\gamma^{\pm}(\lambda,\mu)=-\frac{Q^{\pm}_{0}(\lambda,\mu)}{Q^{\pm}_{1}(\lambda,\mu)},
\end{align*}
where $Q^{\pm}_{0}(\lambda,\mu)$ and $Q^{\pm}_{1}(\lambda,\mu)$ are defined as
\begin{align}\label{polynomial:q_1}
&Q^{\pm}_{0}(\lambda,\mu)=\mp(\lambda+\mu)+\lambda\mu,\\
&Q^{\pm}_{1}(\lambda,\mu)=\mp(\lambda+2\mu)+1+\lambda\mu. \nonumber
\end{align}
\end{lemma}
 
\begin{proof}
It is sufficient to show that  the following system of equations has no solutions  
\begin{equation}\label{eq:system1}
\begin{cases}
Q^{\pm}_{1}(\lambda,\mu)=0\\
C^{\pm}(\gamma,\lambda,\mu)= 0.
\end{cases}
\end{equation}  

System \eqref{eq:system1} obeys that
\begin{equation}\label{eq:system2}
\begin{cases}
Q^{\pm}_{1}(\lambda,\mu)=0\\
Q^{\pm}_{0}(\lambda,\mu)=0
\end{cases} 
\quad \Rightarrow \quad 
\begin{cases}
\mu=0\\
\mu=\pm 1.
\end{cases}
\end{equation}
The last system has no solutions. Therefore,  \eqref{eq:system1} hasn't any solutions too. It leads to the proof of Lemma \ref{lem:function_mu}.
\black
 \end{proof}

\begin{lemma}\label{lem:function_lambda}
The set of points of $\mathbb{R}^2$ satisfying the equation $Q^{\pm}_{1}(\lambda,\mu)= 0$ coincides with the graph of the function
\begin{align}\label{function:lambda}
 \lambda^{\pm}(\mu)=-\frac{2\mu\mp1}{1\mp\mu}.
\end{align}
 \end{lemma}
 
\begin{proof}
The following system of equations
\begin{equation}\label{eq:system3}
\begin{cases}
1\mp \mu=0\\
Q^{\pm}_{1}(\lambda,\mu)= 0
\end{cases}
\end{equation}
is equivalent to \eqref{eq:system2}, which has no solutions, thus, it hasn't any solutions. It yields the proof of Lemma \ref{lem:function_lambda}. 
\black
\end{proof}

%\newpage
The straight lines $\mu=-1$ and $\mu=1$ separate the graph of the functions $\lambda^{-}(\cdot)$ and $\lambda^{+}(\cdot)$ on the  $(\lambda,\mu)$-plane into  two (different) continuous curves $\{\tau^{-}_1, \tau^{-}_2\}$ and $\{\tau^{+}_1, \tau^{+}_2\}$:  
\begin{align*}
&\tau^{-}_1=\{(\lambda,\mu)\in\R^2:
\lambda=-\frac{1+ 2\mu }{1+ \mu},\,\,\mu>-1\},\\
&\tau^{-}_2=\{(\lambda,\mu)\in\R^2:\lambda=-\frac{1+ 2\mu }{1+\mu},\,\,\mu<-1\}
\end{align*}
and
\begin{align*}
&\tau^{+}_1=\{(\lambda,\mu)\in\R^2:
\lambda=\frac{1- 2\mu }{1- \mu},\,\,\mu<1\},\\
&\tau^{+}_2=\{(\lambda,\mu)\in\R^2:\lambda=\frac{1- 2\mu }{1-\mu},\,\,\mu>1\}.
\end{align*}
The curves $\{\tau^{-}_1,\tau^{-}_2\}$  and $\{\tau^{+}_1,\tau^{+}_2\}$ divide the $(\lambda,\mu)$- plane  into the  domains $D^{-}_{1},D^{-}_{2},D^{-}_{3}$  and  $D^{+}_{1},D^{+}_{2},D^{+}_{3}$ of the functions $\lambda^{-}(\cdot)$ and $\lambda^{+}(\cdot)$, respectively:

\begin{align}\label{def:regions_D-}
& D^{-}_1=\{(\lambda,\mu)\in\R^2:
Q^{-}_{1}(\lambda,\mu)>0,\,\, \mu>-1\},\\ \nonumber
&D^{-}_2 =\{(\lambda,\mu)\in\R^2:
Q^{-}_{1}(\lambda,\mu)<0\},\\ \nonumber
&D^{-}_3= \{(\lambda,\mu)\in\R^2:
Q^{-}_{1}(\lambda,\mu)>0,\,\,\mu<-1\} \nonumber
\end{align}
and
\begin{align}\label{def:regions_D+}
& D^{+}_1=\{(\lambda,\mu)\in\R^2:
Q^{+}_{1}(\lambda,\mu)>0,\,\, \mu<1\},\\ \nonumber
&D^{+}_2 =\{(\lambda,\mu)\in\R^2:
Q^{+}_{1}(\lambda,\mu)<0\},\\ \nonumber
&D^{+}_3= \{(\lambda,\mu)\in\R^2:
Q^{+}_{1}(\lambda,\mu)>0,\,\,\mu>1\}. \nonumber
\end{align}

%\vspase{5m}

\begin{figure}\label{1}
 \centering
 \includegraphics[width=1\textwidth]{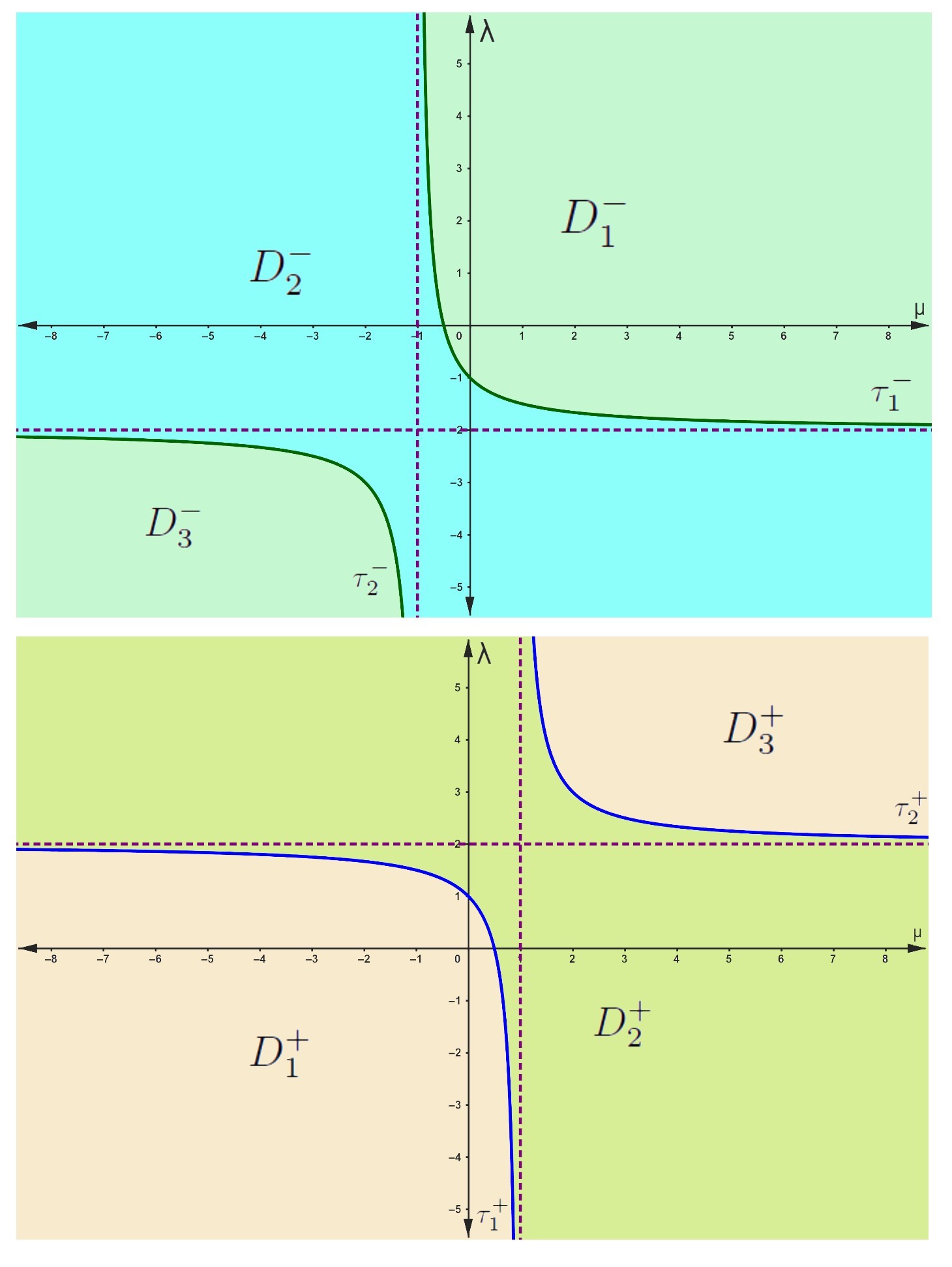}

 \caption{The domains  
$D^{\pm}_{\alpha}, \alpha=1,2,3$ of the principle functions $\gamma^{\pm}$ in the $(\lambda,\mu)$-plane  of the parameters $\lambda,\mu \in \R$.} 
\end{figure}

The following lemma summarizes the locations of the domains $D^{\pm}_{\alpha}, \alpha=1,2,3$ defined in \eqref{def:regions_D-}  and \eqref{def:regions_D+} and also their relations.

\begin{lemma}\label{lem:regions_D}
The followings are true:
\begin{itemize}
\item[(i)]$D^{-}_3 \subseteq D^{+}_1$,
\item[(ii)]$D^{+}_3 \subseteq D^{-}_1$,
\item[(iii)]$D^{-}_3 \cap D^{+}_3=\emptyset$,
\item[(iv)] For any $\alpha=1,2,3$ the regions $D^{-}_\alpha$ and  $ D^{+}_\alpha$ are symmetric with respect to  origin.
\end{itemize}
\end{lemma}
\begin{proof}
The definitions of $D^{\pm}_{\alpha}, \alpha=1,2,3$ yield the proofs of items   (i)-(iii) of Lemma \ref{lem:regions_D} and the equality $Q^{+}_{1}(\lambda,\mu)=Q^{-}_{1}(-\lambda,-\mu)$ yields  (iv)\ (see Figure 1). 
\end{proof}

Recall that, the domain of the function $\gamma^{\pm}(\cdot,\cdot)$  is an open set \black 
\begin{align*}
&\R^2\setminus (\tau^{\pm}_1 \cup \tau^{\pm}_2)=D^{\pm}_{1} \cup D^{\pm}_{2} \cup D^{\pm}_{3}. 
\end{align*}
The curves $\tau^{\pm}_1$ and $ \tau^{\pm}_2$ in $\R^2$ 
are define the corresponding surfaces $\Upsilon^{\pm}_j\subset\R^3,\,j=1,2$\ : 
\begin{align*}
&\Upsilon^{\pm}_j:=\{(\gamma,\lambda,\mu)\in\R^3,\,(\lambda,\mu)\in \tau^{\pm}_j\}.
\end{align*} 

Further, the surfaces $\Upsilon^{\pm}_1$  and $\Upsilon^{\pm}_2$ 
will  separate the graph of the function $\gamma^{\pm}(\cdot,\cdot)$    into three different continuous (connected) surfaces $\Gamma^{\pm}_{1}$, $\Gamma^{\pm}_{2}$ and $\Gamma^{\pm}_{3}$ in $\R^3$:
\begin{align*}
&\Gamma^{\pm}_{j}=\{(\gamma,\lambda,\mu)\in\R^3:\gamma= -\frac{Q^{\pm}_{0}(\lambda,\mu)}{Q^{\pm}_{1}(\lambda,\mu)},\,\, (\lambda,\mu)\in D^{\pm}_j\}, j=1,2,3.
\end{align*}

The surfaces $\Gamma^{-}_{1},\Gamma^{-}_{2},\Gamma^{-}_{3}$ and $\Gamma^{+}_{1},\Gamma^{+}_{2},\Gamma^{+}_{3}$  divides the three dimensional space $\R^3$  into four separated  connected components $\mathbb{G}^{-}_{0}$, $\mathbb{G}^{-}_{1}$, $\mathbb{G}^{-}_{2}$, $\mathbb{G}^{-}_{3}$ and $\mathbb{G}^{+}_{0}$, $\mathbb{G}^{+}_{1}$, $\mathbb{G}^{+}_{2}$, $\mathbb{G}^{+}_{3}$   respectively:
\begin{align}\label{four_sets}
&\mathbb{G}^{-}_{0}:=\{(\gamma,\lambda,\mu)\in\R^3:\,C^{-}(\gamma,\lambda,\mu)>0, \,\,(\lambda,\mu)\in D^{-}_1\},\\
&\mathbb{G}^{-}_{1}:=\{(\gamma,\lambda,\mu)\in\R^3:\,C^{-}(\gamma,\lambda,\mu\}<0,\,\,(\lambda,\mu)\in \overline{D^{-}_1}\cup D^{-}_2\},\nonumber \\
&\mathbb{G}^{-}_{2}:=\{(\gamma,\lambda,\mu)\in\R^3:\,C^{-}(\gamma,\lambda,\mu)>0, \,\,(\lambda,\mu)\in \overline{D^{-}_2}\cup D^{-}_3\},\nonumber\\
&\mathbb{G}^{-}_{3}:=\{(\gamma,\lambda,\mu)\in\R^3:\,C^{-}(\gamma,\lambda,\mu)<0, \,\,(\lambda,\mu)\in D^{-}_3\} \nonumber
\end{align}
and
\begin{align}\label{fourSets}
&\mathbb{G}^{+}_{0}:=\{(\gamma,\lambda,\mu)\in\R^3:\,C^{+}(\gamma,\lambda,\mu)>0, \,\,(\lambda,\mu)\in D^{+}_1\},\\
&\mathbb{G}^{+}_{1}:=\{(\gamma,\lambda,\mu)\in\R^3:\,C^{+}(\gamma,\lambda,\mu\}<0,\,\,(\lambda,\mu)\in \overline{D^{+}_1}\cup D^{+}_2\},\nonumber\\
&\mathbb{G}^{+}_{2}:=\{(\gamma,\lambda,\mu)\in\R^3:\,C^{+}(\gamma,\lambda,\mu)>0, \,\,(\lambda,\mu)\in \overline{D^{+}_2}\cup D^{+}_3\},\nonumber\\
&\mathbb{G}^{+}_{3}:=\{(\gamma,\lambda,\mu)\in\R^3:\,C^{+}(\gamma,\lambda,\mu)<0, \,\,(\lambda,\mu)\in D^{+}_3\}.\nonumber
\end{align}
 
\begin{theorem}\label{teo:constant}
Let $\cC$ be one of the above connected components  of the partition 
of the $(\gamma,\lambda,\mu)$-space. Then for any $(\gamma,\lambda,\mu)\in\cC$ the
numbers 
$n_+({H}_{\gamma\lambda\mu}(0))$ and $n_-({H}_{\gamma\lambda\mu}(0))$ of eigenvalues of
$H_{\gamma\lambda\mu}(0)$ lying,  respectively, below and above the
essential spectrum \eqref{eq:essent_spec}  remain constant. 
\end{theorem}

Now we study the number of eigenvalues of
$H_{\gamma\lambda\mu}(0)$ in $(-\infty,0)$ and $(4,+\infty)$ depending on the potential parameters $\gamma$, $\lambda$ and $\mu$. 

\begin{theorem}\label{teo:number_K=0}
For any $\alpha=0,1,2,3$ the following statements  are true:
\begin{itemize}
\item[(i)] if $(\gamma,\lambda,\mu)\in
\mathbb{G}^{-}_{\alpha}$, then  $n_-({H}_{\gamma\lambda\mu}(0))=\alpha$; 

\item[(ii)] if $(\gamma,\lambda,\mu)\in
\mathbb{G}^{+}_{\alpha}$, then   $n_+({H}_{\gamma\lambda\mu}(0))=\alpha$. 
\end{itemize}
\end{theorem}

%It is turn out that, each group continuous surfaces $\Gamma^{+}_{1},\Gamma^{+}_{2},\Gamma^{+}_{3}$ and $\Gamma^{-}_{1},\Gamma^{-}_{2},\Gamma^{-}_{3}$ divides the three dimensional space $\mathbb{R}^3$  into four disjoint parts.

In the following lemma  the locations of the  connected components 
\begin{equation}\label{def:cC_components}
\cC_{\alpha\beta}=\mathbb{G}_{\alpha}^{-}\cap \mathbb{G}_{\beta}^{+}, \,\, \alpha,\beta=0,1,2,3
\end{equation} are described, which divide the three dimensional space $\mathbb{R}^3$  into a disjoint parts.
\begin{lemma}\label{lem:intersections}
The following relations are true:
\begin{itemize}
\item[(i)]$\cC_{30}=\mathbb{G}_{3}^{-}\cap \mathbb{G}_{0}^{+}=\mathbb{G}_{3}^{-}$;
\item[(ii)]$\cC_{03}=\mathbb{G}_{0}^{-}\cap \mathbb{G}_{3}^{+}=\mathbb{G}_{3}^{+}$;
\item[(iii)] If \ $
\alpha+\beta>3$, then $\cC_{\alpha\beta}=\mathbb{G}_{\alpha}^{-}\cap \mathbb{G}_{\beta}^{+}=\emptyset$. 
\item[(iv)]If \ $
\alpha+\beta\leq 3$, then $\cC_{\alpha\beta}$ is symmetric to  $\cC_{\beta\alpha}$ with respect to  origin  .
\end{itemize}
\end{lemma}

\begin{proof} 

The  assertions (i)-(iii) follow from the definition \eqref{four_sets}, Lemma \ref{lem:regions_D} and several  simple inequalities. Proof of the assertion (iv)  is obtained by the equality  $C^{+}(\gamma,\lambda,\mu)=C^{-}(-\gamma,-\lambda,-\mu)$. 

\end{proof}

Theorem \ref{teo:number_K=0}  and  definition \eqref{def:cC_components} of connected components $\cC_{\alpha\beta}$ yield the following corollary.

\begin{corollary}\label{teo:eigs_below_above}
Let $\gamma,\lambda,\mu\in\mathbb{R}$. If $(\gamma,\lambda,\mu)\in\cC_{\alpha\beta}$ then
$$
n_-({H}_{\gamma\lambda\mu}(0))=\alpha \quad \text{and} \quad n_+({H}_{\gamma\lambda\mu}(0))=\beta,
$$
 where $\alpha,\beta=0,1,2,3 $  and  $\alpha+\beta\leq 3.$

\end{corollary}

Recall that by the min-max principle 
$H_{\gamma\lambda\mu}(K)$ can have at most three eigenvalues outside its
essential spectrum.

The following theorem provides the sharp lower bound for the number of eigenvalues lying outside the essential spectrum of $H_{\gamma\lambda\mu}(K)$ depending only on $\gamma$,$\lambda$ and $\mu$.

\begin{theorem}\label{teo:xosq_kamida}
Let $K\in\T$ and $(\gamma,\lambda,\mu)\in\R^3.$ For all $\alpha,\beta=0,1,2,3$ satisfying the condition $\alpha+\beta\leq 3$ the following relations are true :
\begin{itemize}
\item[(i)] $ \text{if} \quad  (\gamma,\lambda,\mu)\in \cC_{\alpha\beta} \quad \text{and} \quad \alpha+\beta<3,$ 
%\hspace*{-8mm}
$\text{then} \,\, n_{-}(H_{\gamma\lambda\mu}(K)) \ge \alpha,\,\,  n_{+}(H_{\gamma\lambda\mu}(K)) \ge \beta$;  
\item[(ii)] $\text{if} \quad(\gamma,\lambda,\mu)\in \cC_{\alpha\beta} \quad \text{and} \quad \alpha+\beta=3,$ %\hspace*{-8mm}
then \,\, $ n_{-}(H_{\gamma\lambda\mu}(K)) = \alpha,\,\,  n_{+}(H_{\gamma\lambda\mu}(K)) = \beta.$ 
\end{itemize}
\end{theorem}

\section{Proof of the main results}\label{sec:proof}
\subsection{The discrete spectrum of $H_{\gamma\lambda\mu}(0)$}
In the case of $K=0$, the Fredholm determinant
$\Delta_{\gamma\lambda\mu}(0,z)$ is  easier to study.  Note that the essential spectrum of  Hamiltonian
$H_{\gamma\lambda\mu}(0)$
coincides with the segment $[0,4]$.

We try to find an (implicit) equation for the discrete eigenvalues of $H_{\gamma\lambda\mu}(0)$, i.e., for the non-zero solutions of equation
$$
H_{\gamma\lambda\mu}(0)f=zf
$$
in $z\in \R \setminus [0,4]$. 

We apply the Fredholm's determinants method to study the number and location of eigenvalues (see, e.g.,
\cite{Albeverio:1988, Lakaev:1989}). The Fredholm determinant  associated to $H_{\gamma\lambda\mu}(0)$ can be written as
\begin{equation}\label{determinant}
\Delta_{\gamma\lambda\mu}(0;z)= \left|\begin{array}{cc}
1+\gamma a(z)\quad\lambda b(z)\quad \quad\mu c(z)\quad\\
\quad\gamma b(z)\quad 1+\lambda d(z)\quad \mu e(z)\quad\\
\quad\quad\gamma c(z)\quad \lambda e(z)\quad 1+\mu f(z)\quad\\
\end{array}\right|,
\end{equation}
where

\begin{align}\label{a,b,c}
&a(z):=\int\limits_\mathbb{T}\frac{2dt}{\varepsilon (t)-z},
&b(z):=\int\limits_\mathbb{T}\frac{2\cos tdt}{\varepsilon (t)-z},\\
&c(z):=\int\limits_\mathbb{T}\frac{2\cos 2tdt}{\varepsilon (t)-z},
&d(z):=\int\limits_\mathbb{T}\frac{2\cos^2 tdt}{\varepsilon (t)-z}, \nonumber\\
&e(z):=\int\limits_\mathbb{T}\frac{2\cos t\cos 2tdt}{\varepsilon (t)-z},
&f(z):=\int\limits_\mathbb{T}\frac{2\cos^2 {2t}dt}{\varepsilon (t)-z}.\nonumber
\end{align}

Functions $a(\cdot)$, $d(\cdot)$ and $f(\cdot)$ are analytic
in $\R \backslash[0,4]$, strictly decreasing in $\R
\backslash[0,\,4]$, positive in the interval $(-\infty,0)$ and negative in $(4,+\infty)$. The functions $b(\cdot)$, $c(\cdot)$ and $e(\cdot)$ are analytic in $\R \backslash[0,4]$ too. Their behaviour near $z=0$ and $z=4$ are described in the following proposition.

\begin{proposition}\label{asymptotics_of_functions}
The functions  defined in \eqref{a,b,c} have the following asymptotics
\begin{align}
& a(z)=\begin{cases}
\frac{1}{(-z)^{\frac{1}{2}}}+O(-z)^{\frac{1}{2}}, \,\, as \,\, z\nearrow 0  \\ \nonumber
-\frac{1}{(z-4)^{\frac{1}{2}}}+O(z-4)^{\frac{1}{2}}, \,\, as \,\, z\searrow 4  
\end{cases},\\  \nonumber
&b(z)=\begin{cases}
\frac{1}{(-z)^{\frac{1}{2}}}-1+O(-z)^{\frac{1}{2}},\,\, as \,\, z\nearrow 0 \\
-\frac{1}{(z-4)^{\frac{1}{2}}}+1+O(z-4)^{\frac{1}{2}},\,\, as \,\, z\searrow 4
\end{cases}, \\ \nonumber
&c(z)=\begin{cases}
\frac{1}{(-z)^{\frac{1}{2}}}-2+O(-z)^{\frac{1}{2}},\,\, as \,\, z\nearrow 0 \\
-\frac{1}{(z-4)^{\frac{1}{2}}}+2+O(z-4)^{\frac{1}{2}},\,\, as \,\, z\searrow 4 
\end{cases}, \\ \nonumber
&d(z)=\begin{cases}
\frac{1}{(-z)^{\frac{1}{2}}}-1+O(-z)^{\frac{1}{2}},\,\, as \,\, z\nearrow 0 \\
-\frac{1}{(z-4)^{\frac{1}{2}}}+1+O(z-4)^{\frac{1}{2}},\,\, as \,\, z\searrow 4
\end{cases}, \\ \nonumber
&e(z)=\begin{cases}
\frac{1}{(-z)^{\frac{1}{2}}}-2+O(-z)^{\frac{1}{2}}, \,\, as \,\, z\nearrow 0 \\
-\frac{1}{(z-4)^{\frac{1}{2}}}+2+O(z-4)^{\frac{1}{2}}, \,\, as \,\, z\searrow 4
\end{cases}, \\ \nonumber
&f(z)=\begin{cases}
\frac{1}{(-z)^{\frac{1}{2}}}-2+O(-z)^{\frac{1}{2}}, \,\, as \,\, z\nearrow 0 \\ 
-\frac{1}{(z-4)^{\frac{1}{2}}}+2+O(z-4)^{\frac{1}{2}}, \,\, as \,\, z\searrow 4
\end{cases}. \nonumber
\end{align} 
Here $(-z)^{\frac{1}{2}}$ and $(z-4)^{\frac{1}{2}}$ are denote those  branches of analytic functions that are positive for $-z>0$ and $z-4>0$.
\end{proposition}

Proposition \ref{asymptotics_of_functions} can be proved as \cite[Proposition 4]{LBA:2022}. 

\begin{lemma}\label{asym_for_delta}
 For all $\gamma,\lambda,\mu\in\mathbb{R}$  the determinant $\Delta_{\gamma\lambda\mu}(0;z)$ has asymptotics
\begin{align}\label{lemm:asimpdeter}
&\Delta_{\gamma\lambda\mu}(0;z)=\begin{cases}
C^{-}(\gamma,\lambda,\mu)(-z)^{-\frac{1}{2}}+D^{-}(\gamma,\lambda,\mu)+O((-z)^{\frac{1}{2}}),\,\, as \,\, z\nearrow 0, \\
C^{+}(\gamma,\lambda,\mu)(z-4)^{-\frac{1}{2}}+D^{+}(\gamma,\lambda,\mu)+O((z-4)^{\frac{1}{2}}),\,\, as \,\, z\searrow 4,
\end{cases}
\end{align}
where $C^{\pm}$ is defined in \eqref{polynomial:c} and
\begin{align*}
D^{\pm}(\gamma,\lambda,\mu)=1-(\gamma\lambda+2\lambda\mu+4\gamma\mu)\mp(\lambda+\mu+\gamma\lambda\mu).
\end{align*}
\end{lemma}
The proof of Lemma \ref{asym_for_delta} can be derived by using the asymptotics of  functions $a(\cdot), b(\cdot),c(\cdot),d(\cdot),e(\cdot)$ and $f(\cdot)$ in Proposition \ref{asymptotics_of_functions}. 
\begin{corollary}\label{asym_bound}
For all $(\gamma,\lambda,\mu)\in\mathbb{R}^3$ the  relations 
\begin{itemize}
\item[(i).] 
$\lim\limits_{z\rightarrow {-\infty}}\Delta_{\gamma\lambda\mu}(0;z)=1,$
\item[(ii).]
$\lim\limits_{z\nearrow 0}\sqrt{-z}\Delta_{\gamma\lambda\mu}(0;z)=C^{-}(\gamma,\lambda,\mu),$
\item[(iii).]
$\lim\limits_{z\searrow 4}\sqrt{z-4}\Delta_{\gamma\lambda\mu}(0;z)=C^{+}(\gamma,\lambda,\mu)$ 
\end{itemize}
  hold. 
\end{corollary}

\textit{Proof of Corollary \ref{asym_bound}}.
The first item follows from the Lebesgue dominated convergence theorem. Lemma \ref{asym_for_delta} yields the proof of other items.
\hfill   $\square$ \\

The next lemma describes a one-to-one correspondence between the eigenvalues $H_{\gamma\lambda\mu}(0)$ and the zeros of the Fredholm determinant $\Delta_{\gamma\lambda\mu}(0;z)$.
 
\begin{lemma}\label{eig-zero}
The number $z\in \R\setminus[0,4]$ is an eigenvalue of
$H_{\gamma\lambda\mu}(0)$  if
and only if it is a zero of $\Delta_{\gamma\lambda\mu}(0;\cdot).$
Moreover, in $\R\setminus [0,4]$ the function $\Delta_{\gamma\lambda\mu}(0;\cdot)$
has at most three zeros.
\end{lemma}

\begin{proof}
The first assertion follows from the theory of Fredholm determinants (see, for example, \cite{Albeverio:1988}).
Since the operator $H_{\gamma\lambda\mu}(0)$ has rank at most  three, by the minimum-max principle it has at most three eigenvalues outside the essential spectrum. So, according to the first part of the proposition, $\Delta_{\gamma\lambda\mu}(0;\cdot)$ has at most three zeros in $\R\setminus[0,4].$
\end{proof}

The following lemmas determine the number and arrangement of eigenvalues of the operator $H_{\gamma\lambda0}(0)$ that lie below the essential spectrum.

\begin{lemma}\label{simple}
\begin{itemize}
Let $(\gamma,\lambda)\in \mathbb{R}^2$, then the following relations hold.
\item[(\rm{i})]
If $\gamma+\lambda+\gamma\lambda>0$ and $\gamma+1>0$, then $H_{\gamma\lambda0}(0)$ has no eigenvalues in $(-\infty,0)$.
\item[(\rm{ii})]
If $\gamma+\lambda+\gamma\lambda<0$ or $\gamma+\lambda+\gamma\lambda=0$ and $\gamma+1>0$, then $H_{\gamma\lambda0}(0)$ has one eigenvalue in $(-\infty,0)$.
\item[(\rm{iii})]
If $\gamma+\lambda+\gamma\lambda\geq 0$ and $\gamma+1\leq 0$, then $H_{\gamma\lambda0}(0)$ has two eigenvalues in $(-\infty,0)$.
\end{itemize}
\end{lemma}

The following lemma provides the dependence of the number of eigenvalues of the operator $H_{\gamma\lambda0}(0)$ in $(4,+\infty)$ on $\gamma$ and $\lambda$: 
\begin{lemma}\label{simple2}
\begin{itemize}
Let $(\gamma,\lambda)\in \mathbb{R}^2$, then the following relations hold.
\item[(\rm{i})]
If $-\gamma-\lambda+\gamma\lambda>0$ and $\gamma-1<0$, then $H_{\gamma\lambda0}(0)$ has no eigenvalues in $(4,+\infty)$.
\item[(\rm{ii})]
If $-\gamma-\lambda+\gamma\lambda<0$ or $-\gamma-\lambda+\gamma\lambda=0$ and $\gamma-1<0$, then $H_{\gamma\lambda0}(0)$ has one eigenvalue in $(4,+\infty)$.
\item[(\rm{iii})]
If $-\gamma-\lambda+\gamma\lambda\geq 0$ and $\gamma-1\geq 0$, then $H_{\gamma\lambda0}(0)$ has two eigenvalues in $(4,+\infty)$.
\end{itemize}
\end{lemma}

Lemmas \ref{simple} and \ref{simple2} can be proved as in \cite[ Theorem 5.5]{LO'zdemir:2016}.

{\it Proof of Theorem \ref{teo:constant}}. 
Let us assume, without loss of generality, that $C^{-}(\gamma,\lambda,\mu)<0$ for all   $(\gamma,\lambda,\mu)\in \cC$.
The definition of determinant and Lemma \ref{asym_bound} yield  
\begin{equation}\label{ineq:bounds_for_Delta}
\lim\limits_{z\to-\infty} \Delta_{\gamma\lambda\mu}(0,z)  =1, \,
\lim\limits_{z\nearrow 0} \Delta_{\gamma\lambda\mu}(0,z)<0. 
\end{equation}

Let $(\gamma_0,\lambda_0,\mu_0)\in\cC$ be a point of $\cC$ and $z_0<0$ be a zero of the function
$\Delta_{\gamma_0\lambda_0\mu_0}(z)$ of
multiplicity $m\geq 1$. For each fixed $z<0$ the
determinant $\Delta_{\gamma\lambda\mu}(z)$ is a real analytic function in
$(\gamma,\lambda,\mu)\in\cC$ and for each fixed $\gamma,\lambda,\mu \in \R$ the function $\Delta_{\gamma\lambda\mu}(z)$ is
real analytic in $z\in(-\infty,0)$. Hence, for each $\varepsilon>0$ there
are numbers $\delta>0$, $\eta> 0$ and an open neighborhood $W_{\eta}(z_0)$
of $z_0$ with radius $\eta$ such that for all $z\in \overline{W_{\eta}(z_0)}$ and $(\gamma,\lambda,\mu )\in\cC$ obeying the conditions $|z-z_0|=\eta$ and
$||(\gamma,\lambda,\mu)-(\gamma_0,\lambda_0,\mu_0)||<\delta$ the following two
inequalities $|\Delta_{\gamma_0\lambda_0\mu_0}(z)|>\eta$ and
$|\Delta_{\gamma\lambda\mu}(z)-\Delta_{\gamma_0\lambda_0\mu_0}(z)|<\epsilon $ hold. Then by Rouch\'e's theorem the number of zeros of the
function $\Delta_{\gamma\lambda\mu}(z)$ in
${W_{\eta}(z_0)}$ remains constant  for all $(\gamma,\lambda,\mu)\in\cC$ satisfying the inequality
$||(\gamma,\lambda,\mu)-(\gamma_0,\lambda_0,\mu_0)||<\delta $. Since the root
$z_0<0$ of the function $\Delta_{\gamma\lambda\mu}(z)$ is arbitrary in
$(B_1,B_2)$ we conclude that the number of its zeros remains constant in $(B_1,B_2)$ for all
$(\gamma,\lambda,\mu)\in\cC$ satisfying
$||(\gamma,\lambda,\mu)-(\gamma_0,\lambda_0,\mu_0)||<\delta$.

Further each Jordan curve $\Gamma\subset\cC$ connecting any two
points of $\cC$ is a {\it compact set}, so the number of zeros of
the function $\Delta_{\gamma\lambda\mu}(z)$ lying below the bottom of the essential spectrum for any
$(\gamma,\lambda,\mu)\in \Gamma$ remains constant. Therefore, Lemma
\ref{eig-zero} yields that the number of eigenvalues $n_-\bigl(H_{\gamma\lambda\mu}(0)\bigr)$ of the
operator ${H}_{\gamma\lambda\mu}(0)$ below the essential spectrum  is constant.

The proof in the case of $n_+\bigl(H_{\gamma\lambda\mu}(0)\bigr)$
is done in the same way. \hfill $\square$

{\it Proof of Theorem \ref{teo:number_K=0}}.
Let us prove Theorem \ref{teo:number_K=0} in the cases $\alpha=0,1,2,3$ successively.
%\black
According to \eqref{four_sets} and \eqref{polynomial:c} we have $(0,0,1)\in \mathbb{G}^{-}_{0}$.The representation
\eqref{determinant} of the determinant $\Delta_{\gamma\lambda\mu}(0;z)$ yields that
$$\Delta_{001}(0;z)=1+f(z).$$
By definition \eqref{a,b,c} of $f(z)$, for all $z\in(-\infty,0)$ the inequalities $f(z)>0$ and $\Delta_{001}(0;z)>0$ hold, i.e., the
determinant $\Delta_{001}(0;z)$ has no negative zeros in $z\in(-\infty,0)$. Lemma \ref{eig-zero} yields that the operator $H_{001}(0)$ has no eigenvalues below the essential spectrum. Theorem \ref{teo:constant} gives that the operator $H_{\gamma\lambda\mu}(0)$ has no eigenvalues below the essential spectrum for all $(\gamma,\lambda,\mu)\in \mathbb{G}^{-}_{0}$. 

Due to \eqref{four_sets} we have $(0,-1,0)\in\mathbb{G}^{-}_{1}$. Lemma \ref{simple} gives that the operator $H_{0(-1)0}(0)$ has exactly one eigenvalue. Theorem  \ref{teo:constant} yields
that for any $(\gamma,\lambda,\mu)\in \mathbb{G}^{-}_{1}$, the operator $H_{\gamma\lambda\mu}(0)$ has exactly one eigenvalue below the essential spectrum.

We note $(-3,-3,0)\in \mathbb{G}^{-}_{2}$. According to (iii) of Lemma \ref{simple} the operator $ H_{(-3)(-3)0}(0)$ has two eigenvalues. Theorem \ref{teo:constant} yields that for any $(\gamma,\lambda,\mu)\in \mathbb{G}^{-}_{2}$  the operator $H_{\gamma\lambda\mu}(0)$ has two eigenvalues below the essential spectrum.

Now assume that $(\gamma,\lambda,\mu)\in \mathbb{G}^{-}_{3}$. 

Due to \eqref{four_sets} we have
 $$
 C^{-}(\gamma,\lambda,\mu)<0, \,\,(\lambda,\mu)\in D^{-}_3.
$$
Definition \ref{def:regions_D-} and  inclusion $(\lambda,\mu)\in D^{-}_3$ give
\begin{equation}\label{ineq:C-3,2}
\mu<-1, \,\, Q_1^{-}(\lambda,\mu)>0.
\end{equation}
Inequalities \eqref{ineq:C-3,2} and  $C^{-}(\gamma,\lambda,\mu)<0$ guarantee that
\begin{equation}\label{ineq:C-3,3}
\gamma<-\frac{Q_0^{-}(\lambda,\mu)}{Q_1^{-}(\lambda,\mu)}=\frac{\mu+1}{Q_1^{-}(\lambda,\mu)}-1<0.
\end{equation}
%\black
According to  negativity of $\gamma$  the function $\Delta_{\gamma00}(0;z)=1+\gamma a(z)$  is continuous and monotone decreasing in $(-\infty,0)$. \\
By \eqref{polynomial:c} and \eqref{ineq:C-3,3} we have that
$C^{-}(\gamma,0,0)=\gamma<0.$
Corollary \ref{asym_bound} yields that
\begin{equation}
\quad \lim\limits_{z\rightarrow -\infty}\Delta_{\gamma00}(0;z)=1
 \quad \text{and} 
 \lim\limits_{z\nearrow 0}\Delta_{\gamma00}(0;z)=-\infty.
\end{equation}

Since the function $\Delta_{\gamma00}(0;z)$ is continuous and monotone decreasing  in the interval $(-\infty,0)$  it has   exactly one zero $z_{11}$ below the essential spectrum.
 Obviously
\begin{equation}\label{1-ineq}
\begin{aligned}
& \Delta_{\gamma00}(0;z )=1+\gamma a(z )>0 \quad \text{if} \quad z<z_{11},\\
& \Delta_{\gamma00}(0;z )=1+\gamma a(z )<0  \quad \text{if} \quad z>z_{11}.
\end{aligned}
\end{equation}

Observe
\begin{equation}\label{Delta(z_11)<0}
\Delta_{\gamma\lambda0}(0;z_{11})=(1+\gamma a(z_{11}))(1+\lambda d(z_{11}))-\gamma\lambda (b(z_{11}))^2=-\gamma\lambda (b(z_{11}))^2<0.
\end{equation}

According to inequalities \eqref{ineq:C-3,2} we have
$$
C^{-}(\gamma,\lambda,0)=\gamma+\lambda+\gamma\lambda=Q_1(\lambda,\mu)-(\mu+1)>0.
$$

Corollary \ref{asym_bound} and inequality \eqref{Delta(z_11)<0} give
$$
\lim\limits_{z\rightarrow -\infty}\Delta_{\gamma\lambda0}(0;z)=1, \quad  \Delta_{\gamma\lambda0}(0;z_{11})<0  \quad \text{and} \quad 
\lim\limits_{z\nearrow 0}\Delta_{\gamma\lambda0}(0;z)=+\infty.
$$

Thus, the continuous function
$\Delta_{\gamma\lambda0}(0;z)$ 
has at least one zero in each intervals 
$(-\infty,z_{11})$ and $(z_{11},0)$. 
Therefore there exists  real numbers satisfying the inequalities  
\begin{equation}\label{inequality}
z_{21}<z_{11}<z_{22}<0
\end{equation}
such that the following equalities hold:
$$\Delta_{\gamma\lambda0}(0;z_{21})=\Delta_{\gamma\lambda0}(0;z_{22})=0.$$ 
Lemma \ref{eig-zero} and the min-max principle yield that $H_{\gamma\lambda0}(0)$ has  at least two eigenvalues below the essential spectrum. Hence the determinant $\Delta_{\gamma\lambda0}(0;z)$ has exactly two zeros $z_{21}$ and $z_{22}$ in $(-\infty,0)$, which yields
\begin{equation}\label{1.1-ineq}
\begin{aligned}
 &(1+\gamma a(z_{21}))(1+\lambda d(z_{21}))=\gamma\lambda (b(z_{21}))^2>0 , \\
&(1+\gamma a(z_{22}))(1+\lambda d(z_{22}))=\gamma\lambda (b(z_{22}))^2>0.
\end{aligned}
\end{equation}
The relations \eqref{1-ineq}, \eqref{inequality} and \eqref{1.1-ineq} yield
\begin{align*}
 &1+\gamma a(z_{21})>0 \quad \text{and} \quad 1+\lambda d(z_{21})>0 , \\
&1+\gamma a(z_{22})<0 \quad \text{and} \quad 1+\lambda d(z_{22})<0.
\end{align*}
\black 
We note  that $(\gamma,\lambda,\mu)\in \mathbb{G}^{-}_{3}$ and so $C^{-}(\gamma,\lambda,\mu)>0$. Applying Corollary \ref{asym_bound} we have that
$$
\lim\limits_{z\rightarrow -\infty}\Delta_{\gamma\lambda\mu}(0;z)=1
 \quad \text{and} \quad 
 \lim\limits_{z\nearrow 0}\Delta_{\gamma\lambda\mu}(0;z)=-\infty.
$$
Observe that
\begin{align*}
\Delta_{\gamma\lambda\mu}(0;z_{21})=&\Delta_{\gamma\lambda0}(0;z_{21})[1+\mu f(z_{21})]+2\gamma\lambda\mu b(z_{21})c(z_{21})e(z_{21})-\\
&\gamma\mu[1+\lambda d(z_{21})]c^2(z_{21})-\lambda\mu[1+\gamma a(z_{21})](e(z_{21}))^2=\\
&\mu\big[\sqrt{-\gamma[1+\lambda d(z_{21})]}c(z_{21})+\sqrt{-\lambda[1+\gamma a(z_{21})]}e(z_{21})\big]^2<0.
\end{align*}
%&-\gamma\mu[1+\lambda d(z_{21})]c^2(z_{21})-\lambda\mu[1+\gamma a(z_{21})](e(z_{21}))^2+\\
%&2\gamma\lambda\mu b(z_{21})c(z_{21})e(z_{21})=\\
Similarly can be shown
\begin{align*}
&\Delta_{\gamma\lambda\mu}(0;z_{22})=-\mu\big[\sqrt{\gamma[1+\lambda d(z_{22})]}c(z_{22})+\sqrt{\lambda[1+\gamma a(z_{22})]}e(z_{22})\big]^2>0.
\end{align*}
Then the relations
$$
 \lim\limits_{z\rightarrow -\infty}\Delta_{\gamma\lambda\mu}(0;z)=1, \quad  \Delta_{\gamma\lambda\mu}(0;z_{21})<0, \quad   \Delta_{\gamma\lambda\mu}(0;z_{22})>0, \quad  \lim\limits_{z\nearrow 0}\Delta_{\gamma\lambda\mu}(0;z)=-\infty.
$$

yield the existence of zeros $z_{31}$, $z_{32}$  and $z_{33}$ of function $\Delta_{\gamma\lambda\mu}(0;z)$
satisfying the inequalities  $z_{31}<z_{21}<z_{32}<z_{22}<z_{33}<0$ , i.e
$$
\Delta_{\gamma\lambda\mu}(0;z_{31})=\Delta_{\gamma\lambda\mu}(0;z_{32})=\Delta_{\gamma\lambda\mu}(0;z_{33})=0.
$$

Hence the function $\Delta_{\gamma\lambda\mu}(0;z)$ has three  single zeros less than $0$.  Lemma \ref{eig-zero} gives that the operator $H_{\gamma\lambda\mu}(0)$ has three eigenvalues below the essential spectrum.

The proof of item (ii) of the theorem is carried out similarly to the proof of item (i). \hfill $\square$

\subsection{The discrete spectrum of $H_{\gamma\lambda\mu}(K)$}
For every $n\ge1$ define
\begin{equation}\label{enK}
e_n(K;\gamma,\lambda,\mu):= \sup\limits_{\phi_1,\ldots,\phi_{n-1}\in
L^{2}(\T)}\,\,\inf\limits_{{\psi}
\in[\phi_1,\ldots,\phi_{n-1}]^\perp,\,\|{\psi}\|=1}
({H}_{\gamma\lambda\mu}(K){\psi},{\psi})
\end{equation}
and
\begin{equation}\label{EnK}
E_n(K; \gamma,\lambda,\mu):= \inf\limits_{\phi_1,\ldots,\phi_{n-1}\in
L^{2}(\T)}\,\,\sup\limits_{{\psi}
\in[\phi_1,\ldots,\phi_{n-1}]^\perp,\,\|{\psi}\|=1}
({H}_{\gamma\lambda\mu}(K){\psi},{\psi}).
\end{equation}
By the minimax principle, $e_n(K;\gamma,\lambda,\mu)\le \cE_{\min}(K)$ and
$E_n(K;\gamma,\lambda,\mu)\ge \cE_{\max}(K).$ Since, the rank of
$V_{\gamma\lambda\mu}$ does not exceed three, by choosing 
$\phi_1\equiv1,$ $\phi_2(p)=\cos p$ and $\phi_3(p)=\cos 2p$ we immediately see that
$e_n(K;\gamma,\lambda,\mu) = \cE_{\min}(K)$ and $E_n(K;\gamma,\lambda,\mu) =
\cE_{\max}(K)$ for all $n\ge4.$

\begin{lemma}\label{lem:monoton_xos_qiymat}
Let $n\ge1$. Then the maps
$$
K\in\T \mapsto \cE_{\min}(K) - e_n(K;\gamma,\lambda,\mu)
$$
and
$$
K\in\T \mapsto E_n(K;\gamma,\lambda,\mu)-\cE_{\max}(K)  
$$
is non-increasing in $(-\pi,0]$ and non-decreasing in $[0,\pi]$.
\end{lemma}

\begin{proof}
For $\psi\in L^2(\T)$, we consider
$$ ((H_0(K) - \cE_{\min}(K))\psi,\psi)=\int_{\T} 2\cos\tfrac{K}{2}\,\big(1-\cos q\big)|\psi(q)|^2\,\d q. $$
Then, the map $K\in\T\mapsto ((H_0(K) - \cE_{\min}(K))\psi,\psi)$
is non-decreasing in $(-\pi,0]$ and is non-increasing in $[0,\pi].$
Since $V_{\gamma\lambda\mu}$ is independent of $K,$ from the definition of
$e_n(K;\gamma,\lambda,\mu)$, the map $K\in\T\mapsto e_n(K;\gamma,\lambda,\mu) -
\cE_{\min}(K)$ is non-decreasing in $(-\pi,0]$ and is non-increasing
in $[0,\pi].$
\end{proof}

{\it Proof of Theorem \ref{teo:disc_Kvs0}}. For any for any $K\in\T$ and
$m\ge1$ Lemma \ref{lem:monoton_xos_qiymat} gives 
\begin{align}\label{eq:min_eigen0}
&0\le \cE_{\min}(0) - e_m(0;\gamma,\lambda,\mu) \le \cE_{\min}(K)-e_m(K;\gamma,\lambda,\mu),\\
\text{and} \nonumber\\
&0\le E_m(0;\gamma,\lambda,\mu)-
\cE_{\max}(0) \le E_m(K;\gamma,\lambda,\mu) - \cE_{\max}(K).
\end{align}
By assumption $e_n(0;\gamma,\lambda,\mu)$ is an  discrete eigenvalue of
$H_{\gamma\lambda\mu}(0)$ lying below the bottom $\cE_{\min}(K)$. So
$\cE_{\min}(0) - e_n(0;\gamma,\lambda,\mu)>0$ and hence, by
\eqref{eq:min_eigen0} and \eqref{eq:essent_spec},
$e_n(K;\gamma,\lambda,\mu)$ is a discrete eigenvalue of $H_{\gamma\lambda\mu}(K)$
for any $K\in\T.$ Since $e_1(K;\gamma,\lambda,\mu)\le \ldots \le
e_n(K;\gamma,\lambda,\mu)<\cE_{\min}(K),$ it follows that
$H_{\gamma\lambda\mu}(K)$ has at least $n$ eigenvalue below the essential
spectrum. The case of $E_n(K;\gamma,\lambda,\mu)$ is
similar. \hfill $\square$

{\it Proof of Theorem \ref{teo:xosq_kamida}} is obtained by combining Theorem \ref{teo:disc_Kvs0} and Corollary \ref{teo:eigs_below_above}. \hfill $\square$
\black

{\bf Acknowledgments}. The authors acknowledge support of this
research by Ministry of Innovative Development of the Republic of
Uzbekistan (Grant No. FZ–20200929224).

\end{document}